\documentclass[conference,onecolumn]{IEEEtran}

\usepackage[letterpaper,margin=0.75in,footskip=0.75in]{geometry}

\makeatletter
\def\endthebibliography{%
  \def\@noitemerr{\@latex@warning{Empty `thebibliography' environment}}%
  \endlist
}
\makeatother

\IEEEoverridecommandlockouts

\usepackage{csquotes}
\usepackage{cite}
\usepackage[breaklinks=true]{hyperref}
\usepackage[cmex10]{amsmath}
\usepackage{amssymb}
\usepackage{amsthm}
\newtheorem{theorem}{Theorem}
\newtheorem{lemma}{Lemma}

\usepackage{mathtools}
\usepackage{commath}
\usepackage{dsfont}
\usepackage{comment}

\newtheoremstyle{break}
  {\topsep}{\topsep}%
  {\itshape}{}%
  {\bfseries}{}%
  {\newline}{}%
\theoremstyle{break}

\theoremstyle{definition}
\newtheorem{definition}{Definition}

\usepackage{titling}
\usepackage{tikz}
\usetikzlibrary{shapes}

\DeclareMathOperator*{\argmax}{arg\,max}
\DeclareMathOperator*{\argmin}{arg\,min}

\title{Incentive Compatibility in Two-Stage Repeated Stochastic Games}
\author{Bharadwaj Satchidanandan and Munther A. Dahleh}
\date{}

\begin{document}

\maketitle

\begin{abstract}
    We address mechanism design for two-stage repeated stochastic games --- a novel setting using which many emerging problems in next-generation electricity markets can be readily modeled. We introduce a new notion of equilibrium called Dominant Strategy Non-Bankrupting Equilibrium (DNBE) which requires players to make very little assumptions about the behavior of other players to employ their equilibrium strategy. Consequently, a mechanism that renders truth-telling a DNBE could be quite effective in molding real-world behavior along truthful lines. We present a mechanism for two-stage repeated stochastic games that renders truth-telling a Dominant Strategy Non-Bankrupting Equilibrium.
\end{abstract}


\section{Introduction}\label{introduction}

The power system is on the cusp of a revolution. The coming decade could witness increased renewable energy penetration, Electric Vehicle (EV) penetration, EV energy storage integration, demand response programs, etc. These changes have a profound impact on electricity market operations. New mechanisms must be devised to address a variety of important problems that are anticipated to arise in next-generation electricity markets. Most of the existing mechanism design settings are insufficient to model certain crucial features of these problems. To address this, we introduce the setting of Two-Stage Repeated Stochastic Games using which many problems that arise in the context of electricity markets can be readily modeled. The setting is an extension of the one-shot two-stage stochastic game introduced in \cite{Mukund2007} to repeated plays.

At a high level, a two-stage stochastic game, as the name suggests, consists of two stages. In the first stage, the players do not know their valuation functions precisely, but rather only know the probability distribution of their valuation functions. It is only in the second stage of the game that the valuation functions realize. However, the social planner cannot wait until the second stage to decide on the outcome. It is constrained to make certain decisions in the first stage itself based on the probability distribution bids of the players' valuation functions. Once the valuation functions realize and are reported in the second stage, the social planner can make corrections to the first stage outcome by taking certain recourse actions but this comes at a cost \cite{Mukund2007}. 
It is the stochasticity of the players' valuation functions and the prospect for them to misreport both the probability distribution \emph{and} the realization of their valuation functions that preclude the use of classical mechanism design techniques to design efficient and incentive-compatible mechanisms for this setting.

Motivated by applications to electricity markets which operate every day, we consider a setting wherein a two-stage stochastic game is played repeatedly. 
Repeated playing affords the players a large class of strategies that adapt a player's actions to all past observations and inferences obtained therefrom. In other settings such as iterative auctions or dynamic games where a large strategy space of this sort manifests, it typically has an important implication for mechanism design: It may be impossible to obtain truth-telling as a dominant strategy equilibrium \cite{bergemann2019dynamic}. Consequently, in such scenarios, it is common to settle for mechanisms that render truth-telling only a Nash equilibrium, or variants thereof, even though Nash equilibria are known to be poor models of real-world behavior. This is owing to each player having to make overly specific assumptions about the behaviors of the other players in order for them to employ their Nash equilibrium strategy, which they may not make. In general, the lesser the burden of speculation in an equilibrium, the more plausible it is that it models real-world behavior.

Guided by the above maxim, we develop a new notion of equilibrium called the \emph{Dominant Strategy Non-Bankrupting Equilibrium (DNBE)} that requires players to make very little assumptions about the behaviors of the other players for them to employ their equilibrium strategy. Specifically, the only assumption that the players are required to make to play their DNBE strategy is that no player employs a strategy that leads to their own bankruptcy. We make this more precise in Section \ref{problemFormulation}. That the assumption is mild in that it is quite likely to hold in practice needs no belaboring. Consequently, a mechanism that implements a certain desired behavior as a DNBE as opposed to only a Nash equilibrium could be quite effective in molding real-world behavior along the desired lines. 

We then present a mechanism for two-stage repeated stochastic games that renders truth-telling a dominant strategy non-bankrupting equilibrium. The mechanism is individually rational in that every player is guaranteed to accrue a nonnegative utility by truth-telling regardless of what strategies the other players employ. Finally, if every player bids truthfully, then the outcome that the mechanism produces maximizes social welfare. The mechanism is a generalization of the mechanism that we have developed in \cite{satchidanandan2020efficient} for energy storage markets.  

Finally, we apply the mechanism to design an efficient and incentive-compatible demand response market. There are two main takeaways that we wish to highlight for designers of next-generation electricity markets. The first is that there is a need to redesign the ``bidding language" of the day-ahead market. In today's electricity markets, the generators and loads bid their supply and demand functions respectively. However, with the inclusion of demand response providers who may not know exactly in the day-ahead market their ability to reduce consumption the following day, the day-ahead market should allow for bids that are only \emph{probabilistic} in nature. It is only in real time, if and when called upon for demand response, that the demand response providers should be required to disclose their actual costs for curtailing consumption. The theory developed in the paper allows for such probabilistic bids to be submitted to the system operator. Secondly, our results show that ``simple" mechanisms like making payments proportional to the power curtailed by demand response providers, which have been employed in previous demand response trials, are incapable of attaining the optimal social welfare. Significant welfare gains can be obtained by employing carefully-designed mechanisms that take into account the uncertainties of the market participants.

The rest of the paper is organized as follows. Section \ref{problemFormulation} begins with a precise description of a two-stage repeated stochastic game, defines the notion of dominant strategy non-bankrupting equilibrium, and formulates the mechanism design problem. Section \ref{mechanismDesign} develops a mechanism for two-stage repeated stochastic games that guarantees truth-telling to be a dominant strategy non-bankrupting equilibrium. 
Section \ref{applications} describes the application of the results to the design of demand response markets. 
Section \ref{relatedWork} provides an account of related work. Section \ref{conclusion} concludes the paper. 

\noindent\textbf{Notation:} Vectors and sequences are denoted using boldface letters. Given a sequence $\mathbf{x}=\{x(1),x(2),\hdots\},$ we denote by $\mathbf{x}^l$ the segment $\{x(1),\hdots,x(l)\}.$ The hat notation is used to denote bids: Given a variable $x$ that is private to a player, we denote by $\widehat{x}$ the bid that the player submits for $x.$


\section{Problem Formulation}\label{problemFormulation}
A two-stage stochastic game played by $n$ players and consisting of a social planner is described by
\begin{enumerate}
    \item a publicly-known set $\Delta$ known as the type space of the players,
    \item a publicly-known set $\Theta$ of probability mass functions over $\Delta$, known as the supertype space of the players, 
    \item for each $i\in\{1,\hdots,n\},$ a probability distribution $\theta_i\in\Theta$, known as player $i$'s supertype, that is privately known to player $i$ in the first stage of the game, and which it is supposed to report to the social planner in the first stage, 
    \item a set $\mathcal{O}_1$ of first-stage outcomes,
    \item a first-stage decision rule $g^*_1:\Theta^n\to\mathcal{O}_1$ according to which the social planner chooses the first-stage outcome as a function of the players' supertype bids, 
    \item for each $i\in\{1,\hdots,n\},$ player $i$'s type $\delta_i\in\Delta$ that is ``drawn by nature" at random according to $\theta_i$, whose realization is privately observed by player $i$ in the second stage of the game, and which it is supposed to report to the social planner in the second stage,
    \item a set $\mathcal{O}_2$ of second-stage outcomes or ``recourse actions" that the social planner can choose,
    \item a second-stage decision rule $g^*_2:\Theta^n\times\Delta^n\to\mathcal{O}_2$ according to which the social planner chooses the second-stage outcome as a function of the players' type and supertype bids,
    \item a cost function $c:\mathcal{O}_1\times\mathcal{O}_2\to\mathbb{R}$ that specifies for every $(o_1,o_2)\in\mathcal{O}_1\times\mathcal{O}_2,$ the cost incurred by the social planner for choosing the outcome $o_1$ in the first stage and taking the recourse action $o_2$ in the second stage,
    \item for each $i\in\{1,\hdots,n\},$ a valuation function $v_i:\Delta\times\mathcal{O}_1\times\mathcal{O}_2\to\mathbb{R}$ of player $i$ that specifies for every $\delta_i\in\Delta$ and every $(o_1,o_2)\in\mathcal{O}_1\times\mathcal{O}_2,$ the valuation of player $i$ if its type is $\delta_i$ and the social planner chooses the outcomes $o_1$ and $o_2$ in the first and the second stage of the game respectively. 
\end{enumerate}

The first- and second-stage decision rules $(g_1^*,g_2^*)$ that we consider are those that maximize the expected social welfare. To elaborate, 
let $g_1:\Theta^n\to\mathcal{O}_1$ be any first-stage decision rule and $g_2:\Theta^n\times\Delta^n\to\mathcal{O}_2$ be any second-stage decision rule. If the players bid their types and supertypes truthfully, then the expected social welfare that results as a consequence of using the decision rule $(g_1,g_2)$ is $$\mathbb{E}_{\boldsymbol{\delta}\sim\boldsymbol{\theta}}\big[\sum_{i=1}^nv_i(\delta_i,g_1(\boldsymbol{\theta}),g_2(\boldsymbol{\theta},\boldsymbol{\delta}))-c(g_1(\boldsymbol{\theta}),g_2(\boldsymbol{\theta},\boldsymbol{\delta}))\big]=:{W}(\boldsymbol{\theta},g_1,g_2).$$
The goal of the social planner is to maximize the expected social welfare, and so the decision rule $(g_1^*,g_2^*)$ that it employs is
\begin{align}
    (g_1^*,g_2^*)=\argmax_{g_1,g_2}\;{W}(\cdot,g_1,g_2),\label{gStarDefn}
\end{align}
where the maximization is defined in the pointwise sense. 
The social planner computes $g_1^*$ and $g_2^*$ and announces it to the players before the game commences.


The problem that we study is one where a two-stage stochastic game of the above form is played repeatedly on each day $l,$ $l\in\mathbb{Z}_+.$ For ease of exposition, we assume that the supertypes of the players remain the same on all days and it is only their types that differ across days, though this assumption can be relaxed in a straightforward manner. Consequently, for each player $i$, $i\in\{1,\hdots,n\},$ we denote by $\theta_i$ its privately known supertype which remains the same on all days and by $\delta_i(l)$ its privately known type on day $l.$ The sequence $\{\boldsymbol{\delta}(1),\boldsymbol{\delta}(2),\hdots\}$ is assumed to be Independent and Identically Distributed (IID) with $\boldsymbol{\delta}(1)\sim\theta_1\times\hdots\times\theta_n.$

\subsection{First-stage strategy}

On each day $l,$ each player $i$ is required to report its supertype to the social planner in the first stage so that the latter can compute the optimal first-stage outcome. Since the players' supertypes are assumed to remain the same on all days, it suffices for the players to bid their supertypes just once, namely, in the first stage of the game on day $1.$ Owing to strategic reasons that will be clear shortly, the players may not bid their supertypes truthfully, and so we denote by ${\widehat{\theta}_i}$ the supertype bid of player $i$ and by $\sigma_i:\Theta\to\Theta$ the first-stage strategy according to which player $i$ constructs its supertype bid. Therefore, $\widehat{\theta}_i=\sigma_i(\theta_i).$ Once all players submit their supertype bids, the social planner computes the first-stage outcome as $g_1^*(\boldsymbol{\sigma}(\boldsymbol{\theta})),$ where $\boldsymbol{\sigma}(\boldsymbol{\theta})\coloneqq[\sigma_1(\theta_1),\hdots,\sigma_n(\theta_n)].$ 
The game then proceeds to the second stage. 

\subsection{Second-stage bidding policy}

In the second stage on each day $l$, each player $i$ observes the realization of $\delta_i(l)$ which it is supposed to report to the social planner. However, owing to strategic reasons that will become clear shortly, the players may not bid their type realizations truthfully, and so we denote by $\widehat{\delta}_i(l)$ player $i$'s type bid on day $l.$ We allow for the player to construct its type bid on any day $l$ using all information available to it until day $l$, and in accordance with any randomized, history-dependent policy of its choosing. 
Specifically, a second-stage bidding policy $\mu$ of player $i$ is a rule which specifies for each $o_1\in\mathcal{O}_1$ and each $l\in\mathbb{Z}_+,$ a probability transition kernel $\mathbb{P}_{\mu}(\widehat{\delta}_i(l)\big\vert\delta_i^l,\widehat{\delta}_i^{l-1},o_2^{l-1};o_1)$ according to which player $i$ constructs its second-stage bid $\widehat{\delta}_i(l)$ on day $l$ if the first-stage outcome is $o_1$. 
We denote by $\Pi_i$ the set of all second-stage bidding policies available to player $i.$ 

Note that the second stage bidding policy is a \emph{rule} that maps the history of observations available to a player to its second-stage bid. While the outcome of the rule is random owing to the types and second-stage outcomes being random, there is nothing random about the rule itself. Consequently, a player without any loss of generality can choose its second-stage bidding policy right on day $1$ as a function of its supertype. This leads to the notion of the second-stage strategy which is described next. 

\subsection{Second-stage strategy}
A \emph{second-stage strategy} of player $i$ is a function $\pi_i:\Theta\to\Pi_i$ which specifies the second-stage bidding policy that it employs as a function of its private supertype $\theta_i.$ Therefore, $\pi_i(\theta_i)$ is the second-stage bidding policy employed by player $i.$ 


Once all players submit their type bids, the social planner computes the second-stage outcome for day $l$ as $o_2(l)=g_2^*(\boldsymbol{\sigma}(\boldsymbol{\theta}),\widehat{\boldsymbol{\delta}}(l))$. Note that once the players' first-stage and second-stage bidding policies are fixed, a functional relationship is established between the types and the type bids, and all random variables become well-defined. 


\subsection{Strategies and Strategy profiles}

We refer to the composition of the first- and second-stage strategies simply as a \emph{strategy}. I.e, $S_i\coloneqq(\sigma_i,{\pi}_i)$ is referred to as the {strategy} of player $i.$ We denote by $\Lambda_i$ the set of strategies available to player $i.$ Finally, we refer to $\boldsymbol{S}\coloneqq(S_1,\hdots,S_n)$ as the \emph{strategy profile} of the players and denote by $\Lambda$ the set of strategy profiles $\Lambda_1\times\hdots\times\Lambda_n$. 

\subsection{Truthful strategies}

The stochasticity of the player types necessitates the definition of truthful strategy to be weaker than requiring a player to bid its type truthfully on all days. 
\begin{definition}
A strategy $S_i=(\sigma_i,{\pi}_i)$ of player $i$, $i\in\{1,\hdots,n\},$ is \emph{truthful} if 
\begin{enumerate}
    \item[(i)] $\sigma_i(\theta)=\theta$ for every $\theta\in\Theta,$ and 
    \item[(ii)] for every $\theta\in\Theta$ and every $o_1\in\mathcal{O}_1,$ there exists $\mathcal{L}\subseteq\mathbb{Z}_+$ with $\lim_{L\to\infty}\frac{1}{L}\sum_{l=1}^L\mathds{1}_{\{l\in\mathcal{L}\}}=0$ such that for all $l\notin\mathcal{L},$ $$\mathbb{P}_{\pi_i(\theta)}(\widehat{\delta}_i(l)\big\vert\delta_i^l,\widehat{\delta}_i^{l-1},o_2^{l-1};o_1)=\mathds{1}_{\{\widehat{\delta}_i(l)=\delta_i(l)\}}.$$ 
    \end{enumerate}
    
    A strategy profile $(S_1,\hdots,S_n)$ is a \emph{truthful strategy profile} if $S_i$ is truthful for every $i\in\{1,\hdots,n\}.$ 
\end{definition}
\noindent In other words, a strategy $S_i$ is truthful if the supertype bid is truthful and the type bid is truthful ``almost all days." We denote by $\mathcal{T}_i\subset\Lambda_i$ the set of all truthful strategies available to player $i.$

\subsection{Payments and utilities}
The social planner collects a payment from each player at the end of each day that is determined as a function of the bids that they submit until that day. We denote by $p_{i,l}:\Theta_1\times\hdots\times\Theta_n\times\Delta_1^l\times\hdots\times\Delta_n^l\to\mathbb{R}$ the payment rule so that $p_{i,l}(\boldsymbol{\widehat{\theta}},\widehat{\boldsymbol{\delta}}^l)$ specifies the amount that player $i$ should pay on day $l$. The utility accrued by player $i$ is defined as
\begin{align}
    u_{i}(S_i,\boldsymbol{S}_{-i},\boldsymbol{\theta},\boldsymbol{\delta}^\infty)\coloneqq\bigg[\liminf_{L\to\infty}\frac{1}{L}\sum_{l=1}^Lv_i(\delta_i(l),g_1^*(\boldsymbol{\widehat{\theta}}),g_2^*(\boldsymbol{\widehat{\theta}},\boldsymbol{\widehat{\delta}}(l)))-p_{i,l}(\boldsymbol{\widehat{\theta}},\boldsymbol{\widehat{\delta}}^l)\bigg].\label{uiAsymptotic}
\end{align}
Note that a player's utility is a random variable that depends on the realization of the type sequence $\boldsymbol{\delta}^\infty.$

\subsection{Non-Bankrupting strategies} 
As mentioned in Section \ref{introduction}, a ``mild" behavioral assumption, one that is quite likely to hold in practice, is that no player behaves in a manner that might result in its own bankruptcy. This is captured by the notion of a non-bankrupting strategy.

\begin{definition}
A strategy $S_i$ of player $i,$ $i\in\{1,\hdots,n\},$ is  \emph{non-bankrupting} if for all $(\boldsymbol{S}_{-i},\boldsymbol{\theta}),$ $$u_i(S_i,\boldsymbol{S}_{-i},\boldsymbol{\theta},\boldsymbol{\delta}^\infty)>-\infty$$ for all $\boldsymbol{\delta}^\infty,$ except perhaps on a set of probability zero. 

A strategy profile $\boldsymbol{S}=(S_1,\hdots,S_n)$ is \emph{non-bankrupting} if $S_i$ is non-bankrupting for every $i\in\{1,\hdots,n\}.$
\end{definition}
\noindent We denote by $\mathcal{NB}_i$ the set of non-bankrupting strategies of player $i,$ by $\mathcal{NB}_{-i}$ the set of non-bankrupting strategy profiles of all players except player $i,$ and by $\mathcal{NB}$ the set of non-bankrupting strategy profiles of all players.

\subsection{Dominant Strategy Non-Bankrupting Equilibrium}
We are now ready to introduce a notion of equilibrium that is ``slightly" weaker than dominant strategy equilibrium.
\begin{definition}
A strategy profile $\boldsymbol{S}=(S_1,\hdots,S_n)\in\mathcal{NB}$ is a \emph{Dominant Strategy Non-Bankrupting Equilibrium (DNBE)} if for all $i\in\{1,\hdots,n\},$ all $S'_{-i}\in\mathcal{NB}_{-i},$ all $S_i'\in\Lambda_i,$ and all $\boldsymbol{\theta},$
\begin{align}
    u_i(S_i,S'_{-i},\boldsymbol{\theta},\boldsymbol{\delta}^\infty)\geq u_i(S'_i,S'_{-i},\boldsymbol{\theta},\boldsymbol{\delta}^\infty)
\end{align}
for all $\boldsymbol{\delta}^\infty,$ except perhaps on a set of probability zero.
\end{definition}

It is perhaps instructive to contrast DNBE with Dominant Strategy Equilibrium (DSE) and Nash Equilibrium (NE) to gain a better appreciation of the notion. Note that for a strategy profile $\boldsymbol{S}$ to form a Nash equilibrium, it must hold for every $i\in\{1,\hdots,n\}$ that $S_i$ is a best response to $\boldsymbol{S}_{-i}.$ On the other hand, for the strategy profile $\boldsymbol{S}$ to form a DNBE, we must have for all $i\in\{1,\hdots,n\}$ that $S_i$ is a best response not only to $\boldsymbol{S}_{-i}$, but also to all $\boldsymbol{S}'_{-i}\in\mathcal{NB}_{-i}.$ It follows that any dominant strategy non-bankrupting equilibrium is also a Nash equilibrium but not vice-versa. The stronger notion of dominant strategy equilibrium requires for all $i\in\{1,\hdots,n\}$ that $S_i$ is a best response to every $\boldsymbol{S}'_{-i}\in\Lambda_{-i},$ and not just to those in $\mathcal{NB}_{-i}$ as required by DNBE. Hence, any dominant strategy equilibrium is also a dominant strategy non-bankrupting equilibrium. Fig. \ref{figHierarchy} illustrates the hierarchy formed by these equilibrium notions. 


\begin{figure}
\centering
\begin{tikzpicture}

\draw[black, thick] (0,0) rectangle (5,5);

\draw[black, thick] (2.5,2.5) ellipse (2.25cm and 1.6cm);
\draw[black, thick] (2.5,2.5) ellipse (1.7cm and 1cm);
\draw[black, thick] (2.5,2.5) ellipse (1cm and 0.4cm);

\node[align=center] at (2.5,2.5) {\text{DSE}};
\node[align=center] at (2.5,3.15) {\text{DNBE}};
\node[align=center] at (2.5,3.75) {\text{NE}};
\node[align=center] at (2.5,4.5) {\text{Set of Strategy Profiles}};

\end{tikzpicture}
\caption{Hierarchy of equilibrium notions. Any dominant strategy equilibrium is also a dominant strategy non-bankrupting equilibrium, and any dominant strategy non-bankrupting equilibrium is also a Nash equilibrium. }\label{figHierarchy}
\end{figure}
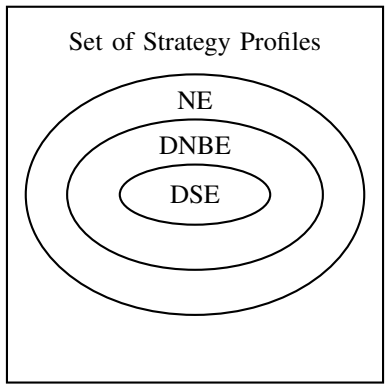

\subsection{Mechanism Design Problem}
Arbitrarily fix the strategy profile $\boldsymbol{S}$ of the players. The long-term average social welfare that results from the game is
\begin{align}
    q(\boldsymbol{S},\boldsymbol{\theta},\boldsymbol{\delta}^\infty)\coloneqq\liminf_{L\to\infty}\frac{1}{L}\sum_{l=1}^L\bigg[\sum_{i=1}^nv_i(\delta_i(l),g_1^*(\boldsymbol{\widehat{\theta}}),g_2^*(\boldsymbol{\widehat{\theta}},\boldsymbol{\widehat{\delta}(l)}))\bigg]-c(g_1^*(\boldsymbol{\widehat{\theta}}),g_2^*(\boldsymbol{\widehat{\theta}},\boldsymbol{\widehat{\delta}}(l))).\label{SWasymptotic}
\end{align}
The objective of the social planner is to ensure that the average social welfare $q(\boldsymbol{S},\boldsymbol{\theta},\boldsymbol{\delta}^\infty)$ equals the optimal value $W^*(\boldsymbol{\theta})$ that would result almost surely if all players employ a truthful strategy. However, the objective of each player $i$ is to maximize its own utility given by (\ref{uiAsymptotic}), and so it may not employ a truthful strategy if there is a possibility for it to accrue a higher utility by doing so than by employing a truthful strategy. 
This brings us to the mechanism design problem. We wish to design a payment rule $\{p_{i,l}: (i,l)\in\{1,\hdots,n\}\times\mathbb{Z}_+\}$ such that each player employing a truthful strategy is a Dominant Strategy Non-Bankrupting Equilibrium. The next section develops the mechanism and establishes the incentive and efficiency properties guaranteed by it.


\section{An Efficient and Incentive-Compatible Mechanism for Two-Stage Repeated Stochastic Games}\label{mechanismDesign}

For each $i\in\{1,\hdots,n\},$ the payment of player $i$ on any day $l$ consists of two components $p_i^{F}$ and $p_i^{S}$ that can be computed by the social planner at the end of the first and the second stages of the game respectively on day $l.$ These payment functions are defined next.

\subsection{First-stage payment}
The first-stage payment $p_i^F$ is a function of only the first-stage bids of the players. Since these quantities remain the same on all days, so do the first-stage payments. The first-stage payment is simply the VCG payment and is defined as
\begin{align}
    p_i^F(\boldsymbol{\widehat{\theta}})\coloneqq W^*(\boldsymbol{\widehat{\theta}}_{-i})-\mathbb{E}_{\boldsymbol{\delta}\sim\mathbb{P}_{\boldsymbol{\widehat{\theta}}}}\bigg[\sum_{j\neq i}v_j(\delta_j,g_1^*(\boldsymbol{\widehat{\theta}}),g_2^*(\boldsymbol{\widehat{\theta}},\boldsymbol{\delta}))-c(g_1^*(\boldsymbol{\widehat{\theta}}),g_2^*(\boldsymbol{\widehat{\theta}},\boldsymbol{\delta}))\bigg],\label{piF}
\end{align}
where $\widehat{\boldsymbol{\theta}}_{-i}$ denotes the supertype bids of all players other than player $i$.

\subsection{Second-stage payment}
At a high level, the first functionality of the second-stage payment is to bind the first-stage and the second-stage strategies of the players. To achieve this, the second-stage payment rule compares the empirical frequencies of the players' type bids with their supertype bids and penalizes discrepancies between them. To elaborate, denote by $\widehat{\theta}_i(t)$ the probability that a random variable distributed according to $\widehat{\theta}_i$ takes the value $t,$ $t\in\Delta.$ On each day $l$ and for each player $i,$ $(l,i)\in\mathbb{Z}_+\times\{1,\hdots,n\},$ the second-stage payment rule computes the discrepancy 
\begin{align}
    \widehat{f}_{i,t}(l)\coloneqq\bigg[\frac{1}{l}\sum_{l'=1}^l\mathds{1}_{\{\widehat{\delta}_i(l')=t\}}\bigg]-\widehat{\theta}_i(t)\label{fhatDefn}
\end{align}
for every $t\in\Delta$, and imposes a penalty of $J_p(l)$ on player $i$ if $\widehat{f}_{i,t}(l)$ falls outside a window of size $r(l)$ for some $t\in\Delta,$ i.e., if
\begin{align}
    \vert\widehat{f}_{i,t}(l)\vert\geq r(l)\label{constraint1}
\end{align}
for some $t\in\Delta.$ 

In a setting of repeated playing, the sequence of second-stage outcomes serves as a source of common randomness which the players can potentially use to correlate their second-stage bids if there is a possibility for them to accrue a higher utility by doing so than by fabricating their bids independently of the other players' bids. The second functionality of the second-stage payment is to disincentivize such strategies. Towards this end, on each day $l$ and for each player $i$, $(l,i)\in\mathbb{Z}_+\times\{1,\hdots,n\},$ the second-stage payment rule computes
\begin{align}
    \widehat{h}_{i,\mathbf{d}}(l)\coloneqq\bigg[\frac{1}{l}\sum_{l'=1}^l\mathds{1}_{\{\widehat{\delta}_i(l')=d_i,\widehat{\boldsymbol{\delta}}_{-i}(l')=\mathbf{d}_{-i}\}}\bigg]-\bigg[\widehat{\theta}_i(d_i)\bigg]\bigg[\frac{1}{l}\sum_{l'=1}^l\mathds{1}_{\{\widehat{\boldsymbol{\delta}}_{-i}(l')=\mathbf{d}_{-i}\}}\bigg]\label{hhatDefn}
\end{align}
for every $\mathbf{d}\in\Delta^{n},$ and imposes a penalty of $J_p(l)$ on player $i$ if it falls outside a window of size $r(l)$ for some $\mathbf{d}\in\Delta^n,$ i.e., if
\begin{align}
    \vert\widehat{h}_{i,\mathbf{d}}(l)\vert\geq r(l)\label{constraint2}
\end{align}
for some $\mathbf{d}\in\Delta^n.$ 

How should the window size sequence $\{r\}$ be chosen? On the one hand, the window size $r(l)$ must tend to zero as $l$ tends to infinity for otherwise, the set of sequences $\{\boldsymbol{\widehat{\delta}}\}$ that satisfy (\ref{constraint1}) and (\ref{constraint2}) would be ``large," thereby violating incentive compatibility. On the other hand, if $\{r\}$ decays too quickly, then even truthful type bids would violate (\ref{constraint1}) and (\ref{constraint2}) infinitely often, thereby incurring a large penalty and violating individual rationality. Hence, the sequence $\{r\}$ should be chosen in a manner that balances the two objectives. This is achieved by choosing $\{r\}$ such that
\begin{align}
    \lim_{l\to\infty}r(l)=0,\label{rGoesToZero}
\end{align}
and for some $\gamma>0,$ 
\begin{align}
    r(l)\geq\sqrt{\frac{\ln{2l^{1+\gamma}}}{2l}}\label{rlgreater}
\end{align}
for all $l\in\mathbb{Z}_+.$\footnote{It suffices that (\ref{rlgreater}) holds not for all $l$ but only for all sufficiently large $l$.}

To obtain an intuition for condition (\ref{rlgreater}), note that the empirical frequency $\frac{1}{l}\sum_{l'=1}^l\mathds{1}_{\{\delta_i(l')=t\}}$ resulting from the true type sequence of player $i$ is a random variable with mean $\theta_i(t)$ and standard deviation that scales as ${{1}/\sqrt{l}}.$ Therefore, if the window size decays at the same rate, then the probability of the empirical frequency falling outside the window would remain at a constant value. This suggests that the window size must scale slower than at least ${{1}/\sqrt{l}}.$ By scaling the window size only slightly slower than ${{1}/\sqrt{l}}$, namely the rate specified by condition (\ref{rlgreater}), truthful bids are guaranteed to almost surely satisfy (\ref{constraint1}) and (\ref{constraint2}) for all but finitely many values of $l$. This is established in Lemma \ref{LemmaHonesty}.

How should the penalty sequence $\{J_p\}$ be chosen? As shown in Lemma \ref{LemmaHonesty}, truthful players incur a penalty only finitely often almost surely, and so the long-term average penalty that they incur is almost surely zero regardless of how the sequence $\{J_p\}$ is chosen. Therefore, the only objective in the design of $\{J_p\}$ is for {every} non-truthful strategy to incur a sufficiently high penalty. This is accomplished by choosing $\{J_p\}$ to be any nonnegative sequence such that 
\begin{align}
    \lim_{l\to\infty} \frac{J_p(l)}{l}=\infty.\label{JpIsOmegaL}
\end{align}

We now have the necessary quantities to define the second-stage payment function. Define the event
\begin{align}
    {E}_{i,\boldsymbol{S}}(l)\coloneqq{\{\max_{t\in\Delta}\;\vert\widehat{f}_{i,t}(l)\vert\geq r(l)\;\cup\; \max_{\mathbf{d}\in\Delta^n}\vert\widehat{h}_{i,\mathbf{d}}(l)\vert\geq r(l)\}}\label{Edefn}
\end{align}
which denotes the occurrence of at least one of (\ref{constraint1}) and (\ref{constraint2}). The second-stage payment of player $i$ on day $l$ is defined as
\begin{align}
    p_{i,l}^S(\boldsymbol{\widehat{\theta}},\boldsymbol{\widehat{\delta}}^l)\coloneqq \Bigg[v_i(\widehat{\delta}_i(l),g_1^*(\boldsymbol{\widehat{\theta}}),g_2^*(\boldsymbol{\widehat{\theta}},\widehat{\boldsymbol{\delta}}(l)))-\mathbb{E}_{\boldsymbol{\delta}\sim\mathbb{P}_{\boldsymbol{\widehat{\theta}}}}\big[v_i(\delta_i,g_1^*(\boldsymbol{\widehat{\theta}}),g_2^*(\boldsymbol{\widehat{\theta}},\boldsymbol{\delta}))\big]\Bigg]+J_p(l)\mathds{1}_{\{E_{i,\boldsymbol{S}}(l)\}}.\label{piS}
\end{align}
A negative value of the above quantity implies a transfer from the social planner to player $i$ on day $l.$ Note that if all players employ a truthful strategy, then the long-term average second-stage payment almost surely equals zero for every player.

The total payment $p_{i,l}(\boldsymbol{\widehat{\theta}},\boldsymbol{\widehat{\delta}}^l)$ that player $i$ transfers to the social planner on day $l$ is
\begin{align}
    p_{i,l}(\boldsymbol{\widehat{\theta}},\boldsymbol{\widehat{\delta}}^l)\coloneqq p_i^F(\boldsymbol{\widehat{\theta}})+p_{i,l}^S(\boldsymbol{\widehat{\theta}},\boldsymbol{\widehat{\delta}}^l).\label{Payment}
\end{align}

The following theorem establishes the incentive and optimality guarantees of the mechanism. 

\begin{theorem}
Consider the two-stage repeated stochastic game induced by the payment rule (\ref{Payment}). 
\begin{enumerate}
    \item A Truthful strategy profile is a dominant strategy non-bankrupting equilibrium. 
    \item If for every $i\in\{1,\hdots,n\}$ and every $\boldsymbol{\theta},$ 
    \begin{align}
        W^*(\boldsymbol{\theta})-W^*(\boldsymbol{\theta}_{-i})\geq 0,\label{noPessimist}
    \end{align} 
    then every player obtains a nonnegative utility by employing a truthful strategy regardless of the strategies that the other players employ. 
    \item If every player employs a truthful strategy, then the long-term average social welfare (\ref{SWasymptotic}) that results is almost surely equal to its optimal value $W^*(\boldsymbol{\theta}).$ 
\end{enumerate}
\end{theorem}

\begin{proof}
Arbitrarily fix $\boldsymbol{\theta},$ $i\in\{1,\hdots,n\},$ the strategy $S_i\in\Lambda_i$ that player $i$ employs, and the strategy profile $\boldsymbol{S}_{-i}\in\mathcal{NB}_{-i}$ that all other players employ. We begin with a lemma. 

\begin{lemma}
For $T_i\in\mathcal{T}_i,$ it holds almost surely that
\begin{align}
    \limsup_{L\to\infty}\frac{1}{L}\sum_{l=1}^L J_p(l)\mathds{1}_{\{E_{i,(T_i,\boldsymbol{S}_{-i})}(l)\}}=0.
\end{align}
I.e., if player $i$ employs a truthful strategy, then the penalty that it pays is almost surely zero.\label{LemmaHonesty}
\end{lemma}
\begin{proof}
It suffices to show that $\{E_{i,(T_i,\boldsymbol{S}_{-i})}(l)\}$ almost surely occurs only finitely often. Arbitrarily fix $\mathbf{d}\in\Delta^n$. Define $\mathcal{F}_l\coloneqq\sigma(\widehat{\boldsymbol{\delta}}_{-i}^l,\delta_i^{l-1})$ so that $\big(\mathds{1}_{\{\widehat{\boldsymbol{\delta}}_{-i}(l')=\mathbf{d}_{-i}\}}\big[\mathds{1}_{\{\delta_i(l')=d_i\}}-\theta_i(d_i)\big],\mathcal{F}_{l'+1}\big)$ is a martingale difference sequence bounded by unity. It follows from the Azuma-Hoeffding inequality that 
\begin{align}
    \mathbb{P}\big(\big\vert\ \frac{1}{l}\sum_{l'=1}^l\mathds{1}_{\{\widehat{\boldsymbol{\delta}}_{-i}(l')=\mathbf{d}_{-i}\}}\big[\mathds{1}_{\{\delta_i(l')=d_i\}}-\theta_i(d_i)\big]\big\vert\geq r(l)\big)\leq 2e^{-2lr^2(l)}.
\end{align}
Combining the above inequality with (\ref{rlgreater}) implies
\begin{align}
    \mathbb{P}\big(\big\vert\ \frac{1}{l}\sum_{l'=1}^l\mathds{1}_{\{\widehat{\boldsymbol{\delta}}_{-i}(l')=\mathbf{d}_{-i}\}}\big[\mathds{1}_{\{\delta_i(l')=d_i\}}-\theta_i(d_i)\big]\big\vert\geq r(l)\big)\leq \frac{1}{l^{1+\gamma}}.
\end{align}
Using (\ref{hhatDefn}) and the fact that player $i$ employs a truthful strategy, the above inequality implies
\begin{align}
    \mathbb{P}\big(\big\vert\widehat{h}_{i,\mathbf{d}}(l)\big\vert\geq r(l)\big)\leq\frac{1}{l^{1+\gamma}}
\end{align}
which in turn implies that $\sum_{l=1}^\infty\mathbb{P}\big(\big\vert\widehat{h}_{i,\mathbf{d}}(l)\big\vert\geq r(l)\big)<\infty.$ Invoking the Borel-Cantelli lemma, we have that $\{\vert\widehat{h}_{i,\mathbf{d}}(l)\vert\geq r(l)\}$ almost surely occurs only finitely often. 

Similarly, $(\mathds{1}_{\{{\delta}_i(l')=d_i\}}-\theta_i(d_i),\mathcal{F}_{l'+1})$ is a martingale difference sequence bounded by unity and following the same sequence of arguments as above, it can be established that $\{\vert\widehat{f}_{i,d_i}(l)\vert\geq r(l)\}$ almost surely occurs only finitely often.

Since $\mathbf{d}$ is arbitrarily chosen, we have that for every $\mathbf{d}\in\Delta^n,$ $\{\vert\widehat{h}_{i,\mathbf{d}}(l)\vert\geq r(l)\}$ and $\{\vert\widehat{f}_{i,d_i}(l)\vert\geq r(l)\}$ almost surely occur only finitely often, and the desired result follows. 
\end{proof}

We have
\begin{align*}
    u_i(S_i,\boldsymbol{S}_{-i},\boldsymbol{\theta},\boldsymbol{\delta}^\infty)&=\liminf_{L\to\infty}\frac{1}{L}\sum_{l=1}^Lv_i({\delta}_i(l),g_1^*(\boldsymbol{\widehat{\theta}}),g_2^*(\boldsymbol{\widehat{\theta}},\widehat{\boldsymbol{\delta}}(l)))-p_{i,l}(\boldsymbol{\widehat{\theta}},\boldsymbol{\widehat{\delta}}^l),
\end{align*}
where $\boldsymbol{\widehat{\theta}}$ and $\boldsymbol{\widehat{\delta}}^\infty$ are determined in accordance with $\boldsymbol{S}$.
Substituting (\ref{piF}) and (\ref{piS}) into (\ref{Payment}), substituting the resulting expression for $p_{i}(l)$ into the above equality, and simplifying the result yields
\begin{align}
    u_i(S_i,\boldsymbol{S}_{-i},\boldsymbol{\theta},\boldsymbol{\delta}^\infty)&=\big[W^*(\boldsymbol{\widehat{\theta}})-W^*(\boldsymbol{\widehat{\theta}}_{-i})\big]\nonumber\\
    &+\liminf_{L\to\infty}\frac{1}{L}\sum_{l=1}^L\bigg(v_i({\delta}_i(l),g_1^*(\boldsymbol{\widehat{\theta}}),g_2^*(\boldsymbol{\widehat{\theta}},\widehat{\boldsymbol{\delta}}(l)))-v_i(\widehat{\delta}_i(l),g_1^*(\boldsymbol{\widehat{\theta}}),g_2^*(\boldsymbol{\widehat{\theta}},\widehat{\boldsymbol{\delta}}(l)))\bigg)\nonumber\\
    &-\limsup_{L\to\infty}\frac{1}{L}\sum_{l=1}^LJ_p(l)\mathds{1}_{\{E_{i,\boldsymbol{S}}(l)\}}.\label{uiAvgExpression}
\end{align}
Arbitrarily fix $T_i\in\mathcal{T}_i$. Then, we obtain using Lemma \ref{LemmaHonesty} and some straightforward algebra that
\begin{align}
    u_i(T_i,\boldsymbol{S}_{-i},\boldsymbol{\theta},\boldsymbol{\delta}^\infty)-u_i(S_i,\boldsymbol{S}_{-i},\boldsymbol{\theta},\boldsymbol{\delta}^\infty)    &=\big[W^*(\theta_i,\boldsymbol{\widehat{\theta}}_{-i})-W^*(\widehat{\theta}_i,\boldsymbol{\widehat{\theta}}_{-i})\big]\nonumber\\
    &+\limsup_{L\to\infty}\frac{1}{L}\sum_{l=1}^L\bigg(v_i(\widehat{\delta}_i(l),g_1^*(\boldsymbol{\widehat{\theta}}),g_2^*(\boldsymbol{\widehat{\theta}},\widehat{\boldsymbol{\delta}}(l)))-v_i({\delta}_i(l),g_1^*(\boldsymbol{\widehat{\theta}}),g_2^*(\boldsymbol{\widehat{\theta}},\widehat{\boldsymbol{\delta}}(l)))\bigg)\nonumber\\
    &+\limsup_{L\to\infty}\frac{1}{L}\sum_{l=1}^LJ_p(l)\mathds{1}_{\{E_{i,\boldsymbol{S}}(l)\}}.\label{utilityDifference}
\end{align}
In what follows, we show that the above quantity is almost surely nonnegative, implying that truthful strategy profiles are Dominant Strategy Non-Bankrupting Equilibria. 

Define 
\begin{align}
    \nu_i(\widehat{\theta}_i,\boldsymbol{\widehat{\theta}}_{-i})\coloneqq\mathbb{E}_{(\widehat{\delta}_i,\widehat{\boldsymbol{\delta}}_{-i})\sim\widehat{\theta}_i\times\boldsymbol{\widehat{\theta}}_{-i}}\bigg[v_i\big(\widehat{\delta}_i,g_1^*(\boldsymbol{\widehat{\theta}}),g_2^*(\boldsymbol{\widehat{\theta}},\widehat{\boldsymbol{\delta}})\big)\bigg]
\end{align}
and
\begin{align}
    \nu_{-i}(\widehat{\theta}_i,\boldsymbol{\widehat{\theta}}_{-i})\coloneqq\mathbb{E}_{(\widehat{\delta}_i,\widehat{\boldsymbol{\delta}}_{-i})\sim\widehat{\theta}_i\times\boldsymbol{\widehat{\theta}}_{-i}}\bigg[\sum_{j\neq i}v_j\big(\widehat{\delta}_j,g_1^*(\boldsymbol{\widehat{\theta}}),g_2^*(\boldsymbol{\widehat{\theta}},\widehat{\boldsymbol{\delta}})\big)-c\big(g_1^*(\boldsymbol{\widehat{\theta}}),g_2^*(\boldsymbol{\widehat{\theta}},\widehat{\boldsymbol{\delta}})\big)\bigg]
\end{align}
so that
\begin{align}
    W^*(\widehat{\theta}_i,\boldsymbol{\widehat{\theta}}_{-i})=\nu_i(\widehat{\theta}_i,\boldsymbol{\widehat{\theta}}_{-i})+\nu_{-i}(\widehat{\theta}_i,\boldsymbol{\widehat{\theta}}_{-i})\label{tempUse1}.
\end{align}
Let $\mu(\Delta,\Delta^n)$ be the set of joint probability mass functions over $\Delta\times\Delta^n.$ For $\psi\in\mu(\Delta,\Delta^n),$ define 
\begin{align}
    \rho_i(\psi)\coloneqq\mathbb{E}_{({\delta}_i,[\widehat{\delta}_i,\widehat{\boldsymbol{\delta}}_{-i}])\sim\psi}\bigg[v_i({\delta}_i,g_1^*(\boldsymbol{\widehat{\theta}}),g_2^*(\boldsymbol{\widehat{\theta}},\widehat{\boldsymbol{\delta}}))\bigg].\label{rhodefn}
\end{align}
Let $\Psi(\theta_i,\boldsymbol{\widehat{\theta}})\subset \mu(\Delta,\Delta^n)$ be the set of joint probability mass functions with ``$x-$marginal" distributed according to $\theta_i$ and ``$y-$marginal" distributed according to ${\widehat{\theta}_1\times\hdots\times\widehat{\theta}_n}.$
{Then, for every $\psi\in\Psi(\theta_i,\boldsymbol{\widehat{\theta}}),$ 
\begin{align}
    W^*(\theta_i,\boldsymbol{\widehat{\theta}}_{-i})\geq\rho_i(\psi)+\nu_{-i}(\widehat{\theta}_i,\boldsymbol{\widehat{\theta}}_{-i}).\label{tempUse2}
\end{align}}
To see this, note that if $(\delta_i,\boldsymbol{\delta}_{-i})\sim\theta_i\times\boldsymbol{\widehat{\theta}}_{-i}$, then the social planner can map $\delta_i$ to a random variable $\delta_i'$ using an appropriate probability transition kernel $P_{\delta_i'\vert\boldsymbol{\delta}}$ such that $(\delta_i,[\delta_i',\boldsymbol{\delta}_{-i}])\sim\psi\in\Psi(\theta_i,\boldsymbol{\widehat{\theta}}).$ Consequently, by choosing the first-stage outcome as $g_1^*(\boldsymbol{\widehat{\theta}})$ and the second-stage outcome as $g_2^*(\boldsymbol{\widehat{\theta}},[{\delta}_i',\boldsymbol{\delta}_{-i}])$, an expected social welfare of $\rho_i(\psi)+\nu_{-i}(\widehat{\theta}_i,\boldsymbol{\widehat{\theta}}_{-i})$ can be attained. It follows that the optimal expected social welfare $W^*(\theta_i,\boldsymbol{\widehat{\theta}}_{-i})$ is at least as large, which yields (\ref{tempUse2})\footnote{This argument requires the second-stage decision rule to be randomized whereas we have assumed $g_1^*$ and $g_2^*$ to be deterministic functions. This apparent gap can be addressed by noting that an optimal decision rule $(g_1^*,g_2^*)$ can be found within the class of deterministic functions.}.

Suppose for a moment that each player $j\in\{1,\hdots,n\}$ employs a stationary second-stage bidding policy $\mu_S^j$ so that $\widehat{\delta}_j(l)$ is chosen as a function of $\delta_j(l)$ according to some probability kernel $P^j_{\widehat{\delta}_j\vert\delta_j}$ for every $l.$ For player $j$'s strategy to be non-bankrupting, it is necessary that $\lim_{L\to\infty}\frac{1}{L}\sum_{l=1}^L \mathds{1}_{\{\widehat{\delta}_j(l)=t\}}=\widehat{\theta}_j(t)$ almost surely for every $t\in\Delta$ for (\ref{constraint1}) would be violated infinitely often otherwise resulting in infinite average penalty. So, for every $j\in\{1,\hdots,n\},$ if player $j$'s strategy is to be non-bankrupting, then $P^j_{\widehat{\delta}_j\vert\delta_j}$ must be such that $\widehat{\delta}_j(1)\sim\widehat{\theta}_j$ given $\delta_j(1)\sim\theta_j$. It follows that for every $j\in\{1,\hdots,n\},$ $(\delta_j(1),\boldsymbol{\widehat{\delta}}(1))\sim\psi_j$ for some $\psi_j\in\Psi(\theta_j,\boldsymbol{\widehat{\theta}})$. It also follows that $\{(\boldsymbol{\delta}(1),\boldsymbol{\widehat{\delta}}(1)),(\boldsymbol{\delta}(2),\boldsymbol{\widehat{\delta}}(2)),\hdots\}$ is a sequence of IID random variables, 
and so we obtain using the Strong Law of Large Numbers (SLLN) that the RHS of (\ref{utilityDifference}) almost surely equals $[W^*(\theta_i,\boldsymbol{\widehat{\theta}}_{-i})-W^*(\widehat{\theta}_i,\boldsymbol{\widehat{\theta}}_{-i})]+[\nu_i(\widehat{\theta}_i,\boldsymbol{\widehat{\theta}}_{-i})-\rho_i(\psi_i)].$ Upon substituting (\ref{tempUse1}), this becomes $W^*(\theta_i,\boldsymbol{\widehat{\theta}}_{-i})-\nu_{-i}(\widehat{\theta}_i,\boldsymbol{\widehat{\theta}}_{-i})-\rho_i(\psi_i),$ and combining it with (\ref{tempUse2}) implies the nonnegativity of (\ref{utilityDifference}). 

However, in order to fabricate the type bids, the players may not restrict just to stationary policies but can employ any history-dependent policy. The rest of the proof is devoted to showing that the same result, namely, the nonnegativity of (\ref{utilityDifference}), holds even in the general case where the players may employ any non-bankrupting strategy. The key to establishing this is the following lemma that characterizes the empirical joint distributions of the reported types when all players employ a non-bankrupting strategy. 

\begin{lemma}
Suppose that for every $j\in\{1,\hdots,n\},$
\begin{align}
    \limsup_{L\to\infty}\frac{1}{L}\sum_{l=1}^LJ_p(l)\mathds{1}_{\{E_{j,\boldsymbol{S}}(l)\}}<\infty.\label{allPlayersNB}
\end{align}
Then, for every $\mathbf{d}\in\Delta^n,$
\begin{align}
    \lim_{L\to\infty}\frac{1}{L}\sum_{l=1}^L\mathds{1}_{\{\boldsymbol{\widehat{\delta}}(l)=\mathbf{d}\}}=\Pi_{j=1}^n\widehat{\theta}_j(d_j).\label{empiricalDistEqualsPhi}
\end{align}\label{lemma1}
\end{lemma}
\begin{proof}
It suffices to show that for all  $\mathbf{d}\in\Delta^n$ and all $k\in\{1,\hdots,n-1\},$
\begin{align}
    \lim_{L\to\infty}\frac{1}{L}\sum_{l=1}^L\mathds{1}_{\{\widehat{\boldsymbol{\delta}}_{k:n}(l)=\mathbf{d}_{k:n}\}}=\widehat{\theta}_k(d_k)\bigg[\lim_{L\to\infty}\frac{1}{L}\sum_{l=1}^L\mathds{1}_{\{\widehat{\boldsymbol{\delta}}_{k+1:n}(l)=\mathbf{d}_{k+1:n}\}}\bigg]\label{l1s1}
\end{align}
and that
\begin{align}
    \lim_{L\to\infty}\frac{1}{L}\sum_{l=1}^L\mathds{1}_{\{\widehat{{\delta}}_{n}(l)={d}_{n}\}}=\widehat{\theta}_n(d_n),\label{l1s2}
\end{align}
where $\mathbf{d}_{k:n}\coloneqq[d_k\;d_{k+1}\;\hdots\;d_n]$ and $\widehat{\boldsymbol{\delta}}_{k:n}(l)$ is defined likewise. 

Combining (\ref{allPlayersNB}) with (\ref{JpIsOmegaL}) implies that $\limsup_{L\to\infty}\sum_{l=1}^L\mathds{1}_{\{E_{j,\boldsymbol{S}}(l)\}}<\infty$ for every $j\in\{1,\hdots,n\}$. I.e., the event sequence $\{E_{j,\boldsymbol{S}}(l)\}$ occurs only finitely often. Hence, we obtain using (\ref{Edefn}) and (\ref{rGoesToZero}) that for all  $\mathbf{d}\in\Delta^n$ and all $j\in\{1,\hdots,n\},$
\begin{align}
    \lim_{L\to\infty}\widehat{f}_{j,d_j}(L)=0\label{fhatlim0}
\end{align}
and 
\begin{align}
    \lim_{L\to\infty}\widehat{h}_{j,\mathbf{d}}(L)=0.\label{hhatlim0}
\end{align}
Substituting (\ref{fhatDefn}) in (\ref{fhatlim0}) implies 
\begin{align}
    \lim_{L\to\infty}\frac{1}{L}\sum_{l=1}^L\mathds{1}_{\{\widehat{\delta}_j(l)=d_j\}}=\widehat{\theta}_j(d_j)\label{marginalEqual}
\end{align}
for all ${d}_j\in\Delta$ and all $j\in\{1,\hdots,n\},$ which in particular establishes (\ref{l1s2}). 

Substituting (\ref{hhatDefn}) in (\ref{hhatlim0}) implies 
\begin{align}
    \lim_{L\to\infty}\frac{1}{L}\sum_{l=1}^L\mathds{1}_{\{\widehat{\delta}_j(l)=d_j,\widehat{\boldsymbol{\delta}}_{-j}(l)=\mathbf{d}_{-j}\}}=\widehat{\theta}_j(d_j)\bigg[\lim_{L\to\infty}\frac{1}{L}\sum_{l=1}^L\mathds{1}_{\{\widehat{\boldsymbol{\delta}}_{-j}(l)=\mathbf{d}_{-j}\}}\bigg]\label{noCollusion}
\end{align}
for all $\mathbf{d}\in\Delta^n$ and all $j\in\{1,\hdots,n\}$. In concluding (\ref{noCollusion}), we have assumed that the limit in the RHS exists, to justify which certain additional arguments are required. We omit these details since they might lessen the focus on the main aspects of the proof. 

The equality (\ref{l1s1}) can now established by noting that
\begin{align}
    \lim_{L\to\infty}\frac{1}{L}\sum_{l=1}^L\mathds{1}_{\{\widehat{\boldsymbol{\delta}}_{k:n}(l)=\mathbf{d}_{k:n}\}}&=\lim_{L\to\infty}\frac{1}{L}\sum_{l=1}^L\sum_{t_1,\hdots,t_{k-1}}\mathds{1}_{\{{\widehat{\delta}}_1(l)=t_1,\hdots,\widehat{\delta}_{k-1}(l)=t_{k-1},\widehat{\boldsymbol{\delta}}_{k:n}(l)=\mathbf{d}_{k:n}\}}\nonumber\\
    &=\sum_{t_1,\hdots,t_{k-1}}\lim_{L\to\infty}\frac{1}{L}\sum_{l=1}^L\mathds{1}_{\{\widehat{\delta}_1(l)=t_1,\hdots,\widehat{\delta}_{k-1}(l)=t_{k-1},\widehat{\boldsymbol{\delta}}_{k:n}(l)=\mathbf{d}_{k:n}\}}\nonumber\\
    &=\sum_{t_1,\hdots,t_{k-1}}\bigg[\widehat{\theta}_k(d_k)\bigg]\bigg[\lim_{L\to\infty}\frac{1}{L}\sum_{l=1}^L\mathds{1}_{\{\widehat{\delta}_1(l)=t_1,\hdots,\widehat{\delta}_{k-1}(l)=t_{k-1},\widehat{\boldsymbol{\delta}}_{k+1:n}(l)=\mathbf{d}_{k+1:n}\}}\bigg]\nonumber\\
    &=\widehat{\theta}_k(d_k)\bigg[\lim_{L\to\infty}\frac{1}{L}\sum_{l=1}^L\sum_{t_1,\hdots,t_{k-1}}\mathds{1}_{\{\widehat{\delta}_1(l)=t_1,\hdots,\widehat{\delta}_{k-1}(l)=t_{k-1},\widehat{\boldsymbol{\delta}}_{k+1:n}(l)=\mathbf{d}_{k+1:n}\}}\bigg]\nonumber\\
    &=\widehat{\theta}_k(d_k)\bigg[\lim_{L\to\infty}\frac{1}{L}\sum_{l=1}^L\mathds{1}_{\{\widehat{\boldsymbol{\delta}}_{k+1:n}(l)=\mathbf{d}_{k+1:n}\}}\bigg],
\end{align}
where the third equality follows from (\ref{noCollusion}). 
\end{proof}

It follows from (\ref{JpIsOmegaL}) that $\limsup_{L\to\infty}\frac{1}{L}\sum_{l=1}^LJ_p(l)\mathds{1}_{\{E_{i,\boldsymbol{S}}(l)\}}$ can only take values $0$ and $\infty.$ In the latter case, the nonnegativity of (\ref{utilityDifference}) is immediate. In the former case, since $\boldsymbol{S}_{-i}$ is a non-bankrupting strategy profile, we have that for all $j\in\{1,\hdots,n\},$ 
\begin{align}
    \limsup_{L\to\infty}\frac{1}{L}\sum_{l=1}^LJ_p(l)\mathds{1}_{\{E_{j,\boldsymbol{S}}(l)\}}<\infty\label{t1u0}
\end{align}
almost surely. Consequently, Lemma \ref{lemma1} applies, and we get
\begin{align}
    \lim_{L\to\infty}\frac{1}{L}\sum_{l=1}^L v_i(\widehat{\delta}_i(l),g_1^*(\boldsymbol{\widehat{\theta}}),g_2^*(\boldsymbol{\widehat{\theta}},\widehat{\boldsymbol{\delta}}(l)))=\nu_i(\widehat{\theta}_i,\boldsymbol{\widehat{\theta}}_{-i}).\label{t1u1}
\end{align}

Now, consider the empirical joint distribution $\psi_L(d,\mathbf{\widehat{d}})\coloneqq\frac{1}{L}\sum_{l=1}^L\mathds{1}_{\{{\delta}_i(l)={d},\widehat{\boldsymbol{\delta}}(l)=\widehat{\mathbf{d}}\}},$ where ${d}\in\Delta$ and $\widehat{\mathbf{d}}\in\Delta^n.$ Note that $\psi_L\in\mu(\Delta,\Delta^n)$ for all $L\in\mathbb{Z}_+.$ It follows from SLLN that for any $d\in\Delta,$ $\lim_{L\to\infty}\sum_{\mathbf{\widehat{d}}}\psi_L(d,\widehat{\mathbf{d}})=\theta_i(d).$ Since (\ref{t1u0}) holds, we obtain using Lemma \ref{lemma1} that for any $\mathbf{\widehat{d}}\in\Delta^n,$ $\lim_{L\to\infty}\sum_{{{d}}}\psi_L(d,\widehat{\mathbf{d}})=\Pi_{j=1}^n\widehat{\theta}_j(\widehat{{d}_j}).$ I.e., the sequence $\{\psi_L\}$ of empirical joint distributions is such that its x-marginal approaches the distribution $\theta_i$ and its y-marginal approaches the distribution ${\widehat{\theta}_1\times\hdots\times\widehat{\theta}_n}.$ It can be shown as a consequence that $\{\psi_L\}$ approaches the set $\Psi(\theta_i,\boldsymbol{\widehat{\theta}})$ in that $\min_{\psi\in\Psi(\theta_i,\boldsymbol{\widehat{\theta}})}\vert\vert\psi-\psi_L\vert\vert\to0$ as $L\to\infty,$ where $\vert\vert\cdot\vert\vert$ can be any norm defined on the set $\mu(\Delta,\Delta^n).$ Also, the function $\rho_i:\mu(\Delta,\Delta^n)\to\mathbb{R}$ defined in (\ref{rhodefn}) is a continuous function over a compact set, and hence uniformly continuous. It follows that 
\begin{align}
    \liminf_{L\to\infty}\rho_i(\psi_L)\leq\sup_{\psi\in\Psi(\theta_i,\boldsymbol{\widehat{\theta}})}\rho_i(\psi).\label{rhoIbounded}
\end{align}
Note also that $\frac{1}{L}\sum_{l=1}^Lv_i(\delta_i(l),g_1^*(\boldsymbol{\widehat{\theta}}),g_2^*(\boldsymbol{\widehat{\theta}},\widehat{\boldsymbol{\delta}}(l)))=\mathbb{E}_{(\delta_i,[\widehat{\delta}_i,\widehat{\boldsymbol{\delta}}_{-i}])\sim\psi_L}[v_i(\delta_i,g_1^*(\boldsymbol{\widehat{\theta}}),g_2^*(\boldsymbol{\widehat{\theta}},\widehat{\boldsymbol{\delta}}))]=\rho_i(\psi_L).$ Taking the limit as $L\to\infty$ and using (\ref{rhoIbounded}) implies
\begin{align}
    \liminf_{L\to\infty}\frac{1}{L}\sum_{l=1}^Lv_i(\delta_i(l),g_1^*(\boldsymbol{\widehat{\theta}}),g_2^*(\boldsymbol{\widehat{\theta}},\boldsymbol{\widehat{\delta}}(l)))\leq\sup_{\psi\in\Psi(\theta_i,\boldsymbol{\widehat{\theta}})}\rho_i(\psi).\label{t1u2}
\end{align}
Substituting (\ref{t1u1}) and (\ref{t1u2}) in (\ref{utilityDifference}) yields
\begin{align*}
    u_i(T_i,\boldsymbol{S}_{-i},\boldsymbol{\theta},\boldsymbol{\delta}^\infty)-u_i(S_i,\boldsymbol{S}_{-i},\boldsymbol{\theta},\boldsymbol{\delta}^\infty)\geq[W^*(\theta_i,\boldsymbol{\widehat{\theta}}_{-i})-W^*(\widehat{\theta}_i,\boldsymbol{\widehat{\theta}}_{-i})]+\nu_i(\widehat{\theta}_i,\boldsymbol{\widehat{\theta}_{-i}})-\sup_{\psi\in\Psi(\theta_i,\boldsymbol{\widehat{\theta}})}\rho_i(\psi).
\end{align*}
Upon substituting (\ref{tempUse1}), the RHS of the above inequality becomes $W^*(\theta_i,\boldsymbol{\widehat{\theta}}_{-i})-\nu_{-i}(\widehat{\theta}_i,\boldsymbol{\widehat{\theta}_{-i}})-\sup_{\psi\in\Psi(\theta_i,\boldsymbol{\widehat{\theta}})}\rho_i(\psi).$ Combining this with (\ref{tempUse2}) implies its nonnegativity, thereby establishing the nonnegativity of (\ref{utilityDifference}).

We now prove the second statement of the theorem. Arbitrarily fix $\boldsymbol{S}_{-i}\in\Lambda_{-i}$ and $T_i\in\mathcal{T}_i.$ Using (\ref{Payment}), (\ref{uiAsymptotic}) and Lemma \ref{LemmaHonesty}, we obtain almost surely that $u_i(T_i,\boldsymbol{S}_{-i},\boldsymbol{\theta},\boldsymbol{\delta}^\infty)=\big[W^*(\theta_i,\boldsymbol{\widehat{\theta}}_{-i})-W^*(\boldsymbol{\widehat{\theta}}_{-i})\big]\geq0$, where the inequality follows from (\ref{noPessimist}). Hence, truth-telling is individually rational for every player.

That the mechanism maximizes social welfare under truthful bidding is a straightforward consequence of the optimality of the first- and the second-stage decision rules. 
\end{proof}

The following section describes an application of the mechanism to the design of demand response markets. 

\section{Application to Demand Response Markets}\label{applications}

As mentioned in Section I, one of the motivating reasons for introducing the environment of a two-stage repeated stochastic game is its ability to readily model many problems that arise in the context of next-generation electricity markets. We illustrate one such problem in this section, namely, mechanism design for demand response markets. In addition to illustrating an application of the proposed framework, the results of this section also serve to illustrate the benefits of using the proposed mechanism as opposed to other ``natural" mechanisms that a policy-maker might employ in such scenarios. 

One of the main requirements of power systems operations is that the power supply has to equal the random demand at each time instant. In conventional systems, the power supply can be controlled, and so the generation is continuously adjusted to follow the random demand to maintain balance. However, at deep levels of renewable energy penetration, the generation becomes random. A popular paradigm for maintaining demand-supply balance in such a system is to make the demand follow the random supply. This typically involves curtailing consumption during times of power supply shortage. This is referred to as demand response, and is achieved by using incentives to modulate the demand.  

One of the key challenges in implementing demand response is that in order to optimally allocate a desired consumption reduction among the demand response providers, their costs for curtailing consumption must be known, which are in general random and private to the loads, and which they could misreport to achieve more favorable allocations for themselves. The goal of the mechanism designer is to elicit both the probability distribution and the realization of the private costs truthfully. See \cite{LCSS2022} for more details. In what follows, we describe how the mechanism developed in the previous section can be applied to this problem. 



In this section, we overload certain notation. Specifically, whenever a demand response market-specific quantity maps to a two-stage repeated stochastic game-specific quantity, the former will be denoted using the same symbol that has been used for the latter. 

Consider a system consisting of $n$ Demand Response (DR) providers and a reserve generator. Each DR provider has a cost function that specifies the cost it incurs as a function of its power consumption reduction. We assume that the cost function is parameterizable and denote by $\delta_i(l)$ the parameter that specifies the cost function of DR provider $i$ on day $l$. Hence, $c(x,\delta_i(l))$ denotes the cost that DR provider $i$ incurs on day $l$ for curtailing its consumption by $x$ units from its baseline. The sequence $\boldsymbol{\delta}^\infty$ is IID with $\boldsymbol{\delta}(1)\sim\boldsymbol{\theta}\coloneqq\theta_1\times\hdots\times\theta_n$ where $\theta_i$ denotes the probability distribution of ${\delta}_i(1).$ The reserve generator has associated with it a production function $c_s:\mathbb{R}\to\mathbb{R}$ which specifies the cost it incurs as a function of the power that it produces. 

Denote by $d(l)$ the power shortage on day $l.$ The system operator wishes to minimize the social cost of compensating the shortage, and therefore wishes to determine the consumption reduction of the DR providers and the reserve generation on day $l$ as

\begin{align}
(\mathbf{x}^*(\boldsymbol{\delta}(l)),g_s^*(\boldsymbol{\delta}(l)))=\argmin_{x_1,\hdots,x_n,g_s} \quad & \sum_{i=1}^nc(x_i,\delta_i(l))+c_s(g_s)\\
\mathrm{subject\; to} \quad & \sum_{i=1}^nx_i+g_s=d(l).\nonumber
\end{align}

The problem of course is that the system operator does not know $\{\delta_1(l),\hdots,\delta_n(l)\},$ and so it requests the DR providers to bid their cost functions. Denote by $\widehat{\delta}_i(l)$ the parameter that DR provider $i$ bids on day $l$. The system operator computes $\mathbf{x}^*(\widehat{\boldsymbol{\delta}}(l))$ and pays each DR load $i$ a payment $p_i(l)$ on day $l$ for reducing its consumption by $x_i^*(\boldsymbol{\widehat{\delta}}(l)).$ The average utility that DR provider $i$ accrues is defined as 
\begin{align}
    u_i^\infty\coloneqq\lim_{L\to\infty}\frac{1}{L}\sum_{l=1}^Lp_i(l)-c(x_i^*(\widehat{\boldsymbol{\delta}}(l)),\delta_i(l)).
\end{align}

\begin{figure}
    \centering
    \includegraphics[scale=0.5]{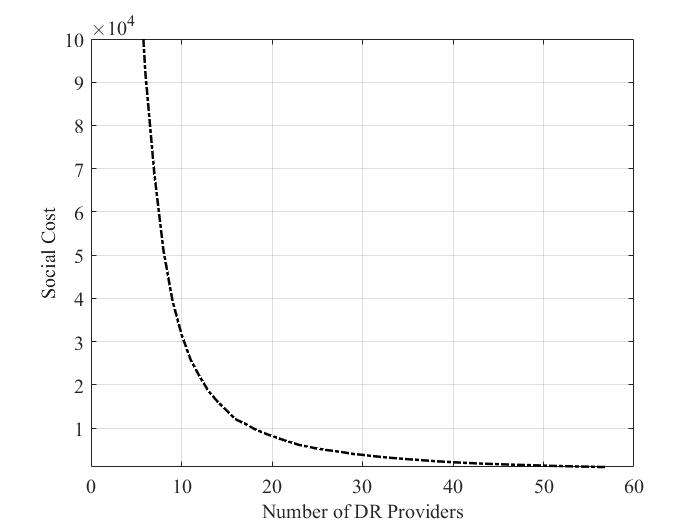}
    \caption{Social cost vs. the number of DR providers. The larger the number of participants in the demand response program, the lower the social cost of the program.}
    \label{fig:socialCost}
\end{figure}

It is straightforward to see that the average utility of each DR provider is a function of not only its own bidding strategy, but also the bidding strategy of the other DR loads. Consequently, a DR provider may not bid its cost truthfully if there is a possibility for it to obtain a larger utility by misreporting its cost. This in turn could result in the demand response program operating in a manner that is not social cost-minimizing. This motivates the mechanism design problem. The mechanism presented in the previous section can be used to design a payment rule which results in truth-telling being a dominant strategy non-bankrupting equilibrium.

For our numerical study, we have taken $c(x,\delta_i)=\frac{\delta_i}{2}x^2$, $c_s(x,\delta_s)=\frac{\delta_s}{2}x^2,$ $\boldsymbol{\theta}$ to be a product of beta distributions of unit mean and variance $2$, and $\delta_s(l)$ to also be beta distributed with the same parameters. 

Fig. \ref{fig:socialCost} quantifies how the social cost reduces as the participation of DR providers increases. Fig. \ref{fig:sensitivity1} illustrates how the payment resulting from the proposed mechanism behaves from the point of view of a randomly chosen DR provider. Specifically, we fix the cost function of a randomly chosen DR provider and plot how its average payment varies with the mean of the costs of the other DR providers. Qualitatively, the higher the mean cost of a DR provider, the higher the inelasticity of its demand. Hence, Fig. \ref{fig:sensitivity1} quantifies the rate at which the payment received by a given DR load increases as a function of the inelasticity of the other DR providers. 

An arguably natural alternative for the proposed mechanism is the posted price mechanism wherein the system operator announces the payment $p_{pp}$ that the DR providers would receive per unit reduction in their power consumption. Each DR provider $i$ then chooses its curtailment $x_{i,pp}^*(l)$ on day $l$ as ${x}_{i,pp}^*(l)=\argmin_{x} \; c(x,\delta_i(l))-p_{pp}x.$ The residual mismatch $d(l)-\sum_{i=1}^nx_{i,pp}^*(l)=:g_s(l)$ is purchased in the spot market at price $c_s(x,\delta_s(l))=\frac{\delta_s(l)}{2}g^2_s(l).$ Such a mechanism has been employed, for example, in a prior demand response trial in the United Kingdom. 

How do such ``simple," ``natural" alternatives compare with the proposed mechanism? Fig. \ref{fig:socialCostComparison} compares the social cost attained by the proposed mechanism with the social cost attained by the posted price mechanism. Certain important observations are in order. First, note that there exists a price point at which the posted price mechanism attains its minimum social cost. However, this price point is a function of the type distributions of the DR loads which are their private knowledge. This necessitates the system operator to perform price discovery in order to compute the optimal price point --- a process that is vulnerable to strategic manipulation by the DR providers. Secondly, even assuming that the DR providers do not manipulate the price discovery, the minimum social cost that can be attained by the posted price mechanism is in general strictly larger than what can be attained by employing the proposed mechanism.

\begin{figure}
    \centering
    \includegraphics[scale=0.5]{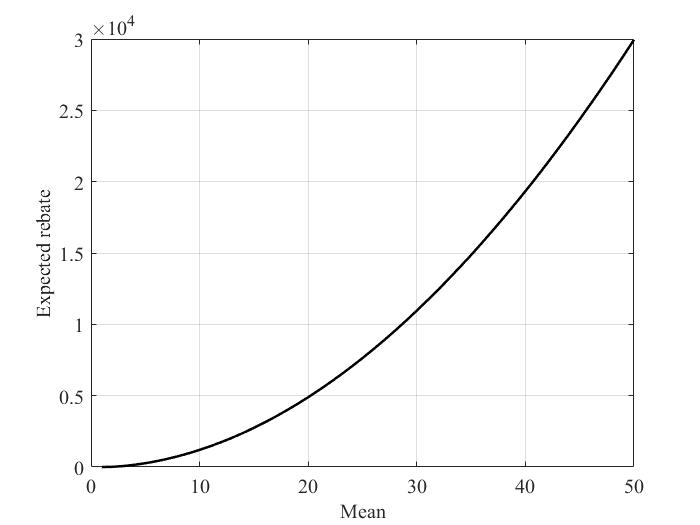}
    \caption{Average payment received by a fixed DR provider as a function of the mean of the supertypes of the other DR providers. The fixed DR provider has cost parameter $\delta_i(l)=4$ for all $l$, and the supertypes of the other loads are beta distributed with varying mean and a fixed variance of $2$. Hence, the average payment received by a given load increases as the demand of the other loads become more and more inelastic.}
    \label{fig:sensitivity1}
\end{figure}


\begin{figure}
    \centering
    \includegraphics[scale=0.5]{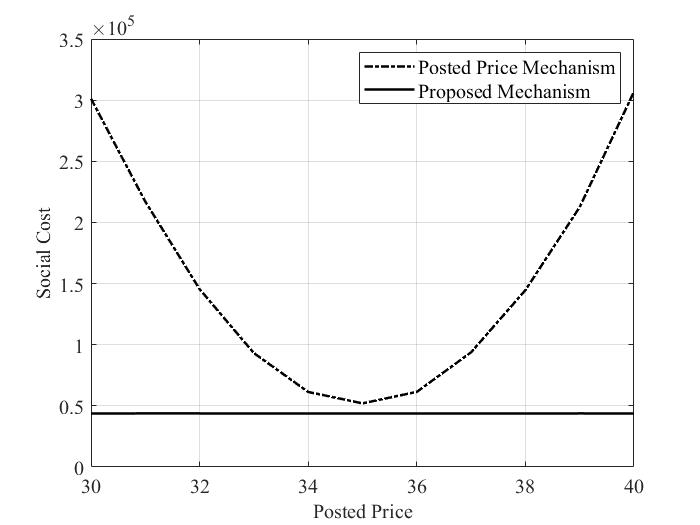}
    \caption{The social cost attained by the posted price mechanism vs. the price.}
    \label{fig:socialCostComparison}
\end{figure}

\section{Related Work}\label{relatedWork}
The setting of two-stage stochastic games was introduced in \cite{Mukund2007} which considers a one-shot setting and develops a mechanism that renders truthful bidding a sequential ex post Nash equilibrium. Reference \cite{Jain1} considers a two-stage game setting to model electricity markets consisting of wind power producers and develops a mechanism that incentivizes truthful bidding. However, it assumes that it is only in the first stage of the game that the wind power producers can bid strategically, and not in the second stage. In contrast, the setting that we have considered assumes that the valuation function distribution \emph{and} the valuation function realization are private to the players, and that they can misreport either or both of them to accrue a higher utility. Reference \cite{mezzetti2004mechanism} presents a two-stage mechanism called the generalized Groves mechanism. In terms of the terminology and the framework presented in this paper, the setting in \cite{mezzetti2004mechanism} can be interpreted as each player having a privately known distribution of its valuation function which it is required to bid to the social planner. The joint distribution of the players' valuation functions is assumed to be common knowledge. The social planner chooses an outcome that maximizes the expected social welfare based on the bids. After the social planner chooses the outcome, the valuation functions realize, which the players are required to bid in the second stage. Following this, a final payment is made. The payment rule guarantees that truth-telling by all players is an ex post Nash equilibrium. It is important to recognize that it is only the payment rule that has two stages in the aforementioned setting, and not the game itself. This in fact is one of the key departures of the one-shot two-stage stochastic game setting from the setting considered in \cite{mezzetti2004mechanism}; the latter doesn't include the possibility for the social planner to take recourse actions after the valuation functions realize. In the context of electricity markets, not only is it feasible to take recourse actions, it is also \emph{imperative} to take recourse actions if grid stability is to be maintained. Reference \cite{BilateralTradeTwoStage} builds upon the mechanism proposed in \cite{mezzetti2004mechanism} to devise a two-stage mechanism for bilateral trade. A power system offering a demand response program is considered in \cite{DRTwoStage,DRTwoStageFull} and a two-stage mechanism is presented using which a certain quantity of power can be apportioned among the loads when a demand response event occurs. The first stage establishes a contingency plan that specifies the amount of power that would be supplied to each load in each contingency and the corresponding price, and the second stage, during which the contingency occurs, allows the loads to trade among themselves at the price established in the first stage. It is shown that the second stage trade results in an allocation that Pareto dominates the first-stage allocation. All of the aforementioned papers consider a one-shot game whereas the setting that we have considered is one of repeated plays. As mentioned in Section \ref{introduction}, the aspect of repeated plays introduces certain additional complexities for mechanism design that can be attributed to the availability history-dependent bidding strategies to players. A similar challenge manifests in dynamic games. References \cite{dynamicMechanism1,dynamicMechanism2,dynamicMechanism3,dynamicAuctions,Ma_Kumar} are some of the papers that address the problem of mechanism design for dynamic games. The solution concept adopted in most of the literature on dynamic games is ex post Nash equilibrium or variants thereof. With the exception of certain special cases such as in \cite{Ma_Kumar}, to the best of our knowledge, we are unaware of any other work that tries to surpass Nash equilibrium or its variants and implement truth-telling in stronger notions of equilibria for broad classes of repeated or dynamic games. A generously disposed view of the present paper could be as an attempt in that direction.

\section{Conclusion}\label{conclusion}
We have considered two-stage repeated stochastic games wherein private information is revealed over two stages and the social planner is constrained to make a decision in each stage. The setting models many important problems that arise in next-generation electricity markets. Recognizing the limitation of Nash equilibria in molding real-world behavior, we have introduced the notion of a dominant strategy non-bankrupting equilibrium which requires the players to make very little assumptions about the behaviors of the other players to employ their equilibrium strategy. Consequently, a mechanism that implements a certain desired behavior as a dominant strategy non-bankrupting equilibrium could effectively mold real-world behavior along the desired lines. We have developed a mechanism for two-stage repeated stochastic games that implements truth-telling as a DNBE. The mechanism is also individually rational and maximizes social welfare.

\bibliographystyle{IEEEtran}
\bibliography{references.bib}

\end{document}